\documentclass{ws-procs9x6}

\usepackage{times}   
\usepackage{microtype}   
\usepackage{amsmath}
\usepackage{amsxtra}
\usepackage{amscd}
\usepackage{amsfonts}
\usepackage{amssymb}
\usepackage{eucal}
\usepackage{epsf}
\usepackage{cite}
\let\a=\alpha \let\be=\beta \let\g=\gamma 
  \let\h=\eta 

 \let\k=\kappa \let\la=\lambda \let\m=\mu
\let\n=\nu \let\x=\xi \let\p=\pi \let\r=\rho \let\s=\sigma
\let\om=\omega 
   \let\Ps=\Psi
  
\let\La=\Lambda  \let\D=\Delta

\let\qd=\quad  

\def\epp{\, .}
\def\epc{\, ,}

\def\2{\frac{1}{2}} \def\4{\frac{1}{4}}

\def\6{\partial}

\def\+{\dagger}

\def\<{\langle} \def\>{\rangle}

\def\i{{\rm i}}

\def\rd{{\rm d}}
\def\re{{\rm e}}

\DeclareMathOperator{\sh}{sh}
\DeclareMathOperator{\ch}{ch}

\DeclareMathOperator{\cth}{cth}

\def\Re{{\rm Re\,}} \def\Im{{\rm Im\,}}

\def\fa{\mathfrak{a}}

\renewcommand{\appendix}{%
   \renewcommand{\section}{%\newpage%
        \secdef\Appendix\sAppendix}%
   \setcounter{section}{0}%
   \renewcommand{\thesection}{\Alph{section}}%
   \renewcommand{\theequation}{\thesection.\arabic{equation}}%
}

\newcommand{\Appendix}[2][?]{%
     \refstepcounter{section}%
     \setcounter{equation}{0}%
     \addcontentsline{toc}{appendix}%
          {\protect\numberline{\appendixname~\thesection} #1}%
     \vspace{\baselineskip}%
     {\noindent\Large\bfseries\appendixname\ \thesection: #2\par}%
     \sectionmark{#1}\vspace{\baselineskip}}

\newcommand{\sAppendix}[1]{%
     {\noindent\large\bfseries\appendixname\:: #1\par}%
     \sectionmark{#1}\vspace{\baselineskip}}

\renewcommand{\tilde}{\widetilde}

\begin{document}

%\title{DRESSED CHARGE AND Q-FUNCTIONS}
\title{Properties of linear integral equations related to the six-vertex
model with disorder parameter}

\author{Hermann Boos and Frank G{\"o}hmann}
\address{Fachbereich C -- Physik, Bergische Universit\"at Wuppertal,\\
42097 Wuppertal, Germany}

\begin{abstract}
One of the key steps in recent work on the correlation functions of
the XXZ chain was to regularize the underlying six-vertex model by a
disorder parameter $\alpha$. For the regularized model it was shown
that all static correlation functions are polynomials in only two
functions. It was further shown that these two functions can be
written as contour integrals involving the solutions of a certain type of
linear and non-linear integral equations. The linear integral equations
depend parametrically on $\alpha$ and generalize linear integral equations
known from the study of the bulk thermodynamic properties of the model. In
this note we consider the generalized dressed charge and a generalized
magnetization density. We express the generalized dressed charge as a
linear combination of two quotients of $Q$-functions, the solutions of
Baxter's $t$-$Q$-equation. With this result we give a new proof of a lemma
on the asymptotics of the generalized magnetization density as a function
of the spectral parameter.
\end{abstract}


\bodymatter

\section{Introduction}
In our present understanding of the thermodynamics \cite{Kluemper93}
and the finite temperature correlation functions\cite{GKS04a,%
BJMST08a,JMS08,BoGo09} of the XXZ quantum spin chain certain complex
valued functions defined as solutions of linear or non-linear integral
equations play an important role. In first place we have to mention the
so-called auxiliary function $\fa$, satisfying the non-linear integral
equation
\begin{equation} \label{nlietemhom}
     \ln (\fa (\la| \k)) = - 2\k \h - \frac{2J \sh(\h) \re (\la)}{T}
        - \int_{C} \frac{\rd \m}{2 \p \i}
		        K(\la - \m) \ln (1 + \fa (\m| \k )) \epp
\end{equation}
\enlargethispage{1ex}
Here $J$ sets the energy scale of the spin chain, $T$ is the temperature,
and $\h$ controls the anisotropy\footnote{The anisotropy parameter of
the XXZ Hamiltonian is $\D = \ch (\h)$ and the quantum group parameter
$q = \re^\h$. For simplicity we shall assume throughout that $\Re \h = 0$
and $0 < \Im \h < \p/2$. This means to consider the XXZ chain in the
critical regime.}. The bare energy $\re (\la)$ and the kernel $K(\la)$ are
defined as
\begin{equation}
     \re(\la) = \cth(\la) - \cth(\la + \h) \epc \qd
     K(\la) = \cth(\la - \h) - \cth(\la + \h) \epp
\end{equation}
The integration contour $C$ encircles the real axis at a distance
slightly smaller than $\g/2 = \Im \h/2$. The twist parameter $\k$ is
proportional to the magnetic field $h$, $\k = h/2T\h$.

The auxiliary function $\fa$ determines the free energy per lattice site,
\begin{equation} \label{freee}
     f (h,T) = - \frac{h}{2} - T \int_{C} \frac{\rd \la}{2 \p \i} \,
                 \re (\la) \ln (1 + \fa (\la| \k)) \epc
\end{equation}
of the spin chain and, hence, all its thermodynamic properties. This
explains the importance of $\fa$.

The magnetization, for instance, is defined as
\begin{equation}
     m(h,T) = - \, \frac{\6 f(h,T)}{\6 h} \epp
\end{equation}
It has a simple expression in terms of the logarithmic derivative of
the auxiliary function,
\begin{equation} \label{defsigma}
     \s (\la) = - T \6_h \ln (\fa (\la|\k)) \epc
\end{equation}
namely,
\begin{equation} \label{msdens}
     m(h,T) = - \2 - \int_{C} \frac{\rd \la}{2 \p \i} \,
                \frac{\re (- \la) \s (\la)}{1 + \fa(\la|\k)} \epp
\end{equation}
The function $\s$ satisfies the linear integral equation
\begin{equation} \label{intsigma}
     \s (\la) = 1 + \int_{C} \frac{\rd \m}{2 \p \i}
                \frac{K(\la - \m) \s (\m)}{1 + \fa (\m|\k)} \epp
\end{equation}
Its zero temperature limit
\begin{equation}
     \x (\la) = \lim_{T \rightarrow 0+} \s(\la)
\end{equation}
is called the dressed charge. It plays an important role in the
calculation of the asymptotics of correlation functions at $T = 0$.
For the lack of any better name we shall call $\s$, and also an
$\a$-generalization of $\s$ to be considered below, the dressed charge as
well.

Another possibility of expressing the magnetization per lattice site
(\ref{msdens}) is by means of a magnetization density $G$ satisfying
\begin{equation} \label{intg}
     G(\la) = \re(-\la) + \int_C \frac{\rd \m}{2 \p \i}
              \frac{K(\la - \m) G(\m)}{1 + \fa (\m|\k)} \epp
\end{equation}
Applying the `dressed function trick' to (\ref{intsigma}) and (\ref{intg})
we obtain
\begin{equation} \label{msdens2}
     m(h,T) = - \2 - \int_{C} \frac{\rd \la}{2 \p \i} \,
                \frac{G (\la)}{1 + \fa(\la|\k)} \epp
\end{equation}

The integrability of the XXZ chain manifests itself in the existence
of commuting families of transfer matrices and $Q$-operators of the
associated six-vertex model \cite{Babook}. With an appropriate staggered
choice of the horizontal spectral parameters the partition function of the
six-vertex model on a rectangular lattice approximates the partition
function of the XXZ chain \cite{Kluemper93}. The approximation becomes exact
in the so-called Trotter limit, when the extension of the lattice in
vertical direction goes to infinity. By a modification of the boundary
conditions in vertical direction we can also obtain an expression for the
density matrix of a finite segment of the spin chain \cite{GKS04a,GHS05}.
The column-to-column transfer matrix in this approach is called the quantum
transfer matrix. It satisfies a $t$-$Q$-equation as well, which becomes a
functional equation for the eigenvalues of the involved operators due
to their commutativity.

We denote the dominant eigenvalue of the quantum transfer matrix
by $\La(\la|\k)$. This eigenvalue alone determines the free energy
in the thermodynamic limit, when the horizontal extension of the
lattice tends to infinity, $f (h,T) = - T \ln \La (0|\k)$. Let the
corresponding $Q$-function be $Q(\la|\k)$. Then $\La$ and $Q$ satisfy
the $t$-$Q$-equation
\begin{equation} \label{tq}
     \La (\la|\k) Q(\la|\k) =
        q^\k a(\la) Q(\la - \h|\k) + q^{-\k} d(\la) Q(\la + \h|\k) \epc
\end{equation}
where $a(\la)$ and $d(\la)$ are the pseudo vacuum eigenvalues of the
diagonal entries of the monodromy matrix associated with the quantum
transfer matrix,
\begin{equation}
\label{detpara}
     a(\la) = \biggl( \frac{\sh(\la + \frac{\be}{N})}
	                     {\sh(\la + \frac{\be}{N} - \h)}
		\biggr)^{\mspace{-6mu} \frac{N}{2}} \epc \qd
     d(\la) = \biggl( \frac{\sh(\la - \frac{\be}{N})}
	                     {\sh(\la - \frac{\be}{N} + \h)}
		\biggr)^{\mspace{-6mu} \frac{N}{2}} \epc
\end{equation}
and $\be = 2 J \sh (\h)/T$.

Using the $Q$-functions corresponding to the dominant eigenvalue the
auxiliary function $\fa$ can be expressed as
\begin{equation} \label{aqq}
     \fa(\la|\k) = \frac{q^{-2 \k} d(\la) Q(\la + \h|\k)}
                        {a(\la) Q(\la - \h|\k)} \epp
\end{equation}
In fact, the auxiliary function $\fa$ is usually defined by
(\ref{aqq}), and afterwards it is shown that $\fa$ satisfies the
non-linear integral equation (\ref{nlietemhom}) in the Trotter limit.
To be more precise, the $Q$-functions, the transfer matrix eigenvalue
and the vacuum expectation values depend implicitly on the Trotter
number $N$. Hence, $\fa$ as defined in (\ref{aqq}) depends on $N$. One
can show\footnote{For a recent pedagogical review on quantum spin chains
within the quantum transfer matrix approach see \cite{GoSu10}, submitted
to the same Festschrift volume for T. Miwa as this article.} that it
satisfies the non-linear integral equation
\begin{multline} \label{nlien}
     \ln \fa (\la|\k) = - 2 \k \h \\ + \ln \biggl[
                      \frac{\sh(\la - \frac{\be}{N})
		            \sh(\la + \frac{\be}{N} + \h)}
			   {\sh(\la + \frac{\be}{N})
			    \sh(\la - \frac{\be}{N} + \h)}
			    \biggr]^\frac{N}{2}
        - \int_{C} \frac{\rd \m}{2 \p \i} \,
          K(\la - \m) \ln (1 + \fa (\m|\k)) \epp
\end{multline}
Clearly this turns into (\ref{nlietemhom}) for $N \rightarrow \infty$.
The integral equation (\ref{nlietemhom}) is the reason why the function
$\fa$ is more useful for practical purposes than $Q$. It is hard to
determine $Q$, and $Q$ has no simple Trotter limit. On the other hand,
(\ref{nlietemhom}) determines $\fa$ directly in the Trotter limit and can
be converted into a form that can be accurately solved numerically.

Inserting (\ref{aqq}) into (\ref{defsigma}) we obtain an expression
for the dressed charge function in terms of logarithmic derivatives
of $Q$-functions.
\begin{equation} \label{sigmazero}
     \s(\la) = 1 + \frac{1}{2 \h} \biggl(
                  \frac{Q'(\la - \h|\k)}{Q(\la - \h|\k)} -
                  \frac{Q'(\la + \h|\k)}{Q(\la + \h|\k)} \biggr) \epc
\end{equation}
where the prime denotes the derivative with respect to $\k$.
For the function $G$ defined in (\ref{intg}) no such simple expression
in terms of $Q$-functions is known.

Below we shall introduce generalizations of the functions $\s$ and $G$
that depend on additional parameters. For the generalized dressed
charge we will derive a generalization of (\ref{sigmazero}). This will
be used in a derivation of the asymptotic behaviour of the generalized
magnetization density as a function of the spectral parameter.
\section{Linear integral equations}
It was shown in \cite{JMS08} that all correlation functions of the
XXZ chain regularized by a disorder parameter $\a$ can be expressed
in terms of two functions, the ratio of eigenvalues
\begin{equation}
     \r (\la) = \frac{\La (\la|\k + \a)}{\La (\la|\k)}
\end{equation}
and a function $\om$ with the essential part $\Ps (\la, \m)$ that can be
characterized in terms of solutions of certain $\a$-dependent linear
integral equations \cite{BoGo09}. A thorough understanding of these two
functions is of fundamental importance for the further study of the
correlation functions of the XXZ chain and for the application of the
lattice results to quantum field theory in various scaling limits
\cite{BJMS09bpp}.

We define the `measure'
\begin{equation}
     \rd m(\la)
        = \frac{\rd \la}{2 \p \i \, \r(\la) (1 + \fa (\la| \k))}
\end{equation}
and the deformed kernel
\begin{equation}
     K_\a (\la) = q^{- \a} \cth (\la - \h) - q^\a \cth (\la + \h) \epp
\end{equation}
Then, for $\n$ inside $C$ the function $G$ is, by definition, the solution
of the integral equation
\begin{multline} \label{newg}
     G(\la, \n) = \\ q^{-\a} \cth(\la - \n - \h) - \r (\n) \cth (\la - \n)
                   + \int_{C} \rd m(\m) K_\a (\la - \m) G(\m, \n) \epp
\end{multline}
Clearly $G$ is a generalization of the magnetization density (\ref{intg})
that depends on an additional spectral parameter and on the disorder
parameter $\a$. For simplicity we keep the same notation also for the
generalized function. $G$ enters the definition of $\Ps$ which,
for $\n_1, \n_2$ inside $C$, is defined as
\begin{equation} \label{Psi}
     \Ps (\n_1, \n_2) = \int_{C} \rd m(\m) G(\m, \n_2)
        \bigl(q^\a \cth(\m - \n_1 - \h) - \r(\n_1) \cth(\m - \n_1) \bigr)
	   \epp
\end{equation}
For $\n$ or $\n_1$, $\n_2$ outside the contour, $G$ and $\Ps$ are given
by the analytic continuations of the right hand side of (\ref{newg}) or
(\ref{Psi}), respectively.

\begin{lemma}
Asymptotic behaviour of $G$ and $\Ps$ as a functions of the spectral
parameters.
\begin{enumerate}
\item
\begin{equation} \label{limg}
     \lim_{\Re \la \rightarrow \infty} G(\la, \n) = 
     \lim_{\Re \n \rightarrow \infty} G(\la, \n) = 0 \epp
\end{equation}
\item
\begin{subequations}
\begin{align} \label{asymppsi1}
     & \lim_{\Re \n_1 \rightarrow \infty} \Ps (\n_1, \n_2) =
        - \frac{q^{- \a} - \r (\n_2)}{1 + q^{2 \k}} \epc \\[1ex]
     & \lim_{\Re \n_2 \rightarrow \infty} \Ps (\n_1, \n_2) =
        - \frac{q^\a - \r (\n_1)}{1 + q^{- 2 \k}} \epp \label{asymppsi2}
\end{align}
\end{subequations}
\end{enumerate}
\end{lemma}

\begin{proof}
We may choose the contour $C$ as the rectangular contour of hight
slightly less than $\g$ and of width $2R$ depicted in
figure~\ref{fig:kontur}. $R$ must be sufficiently large to include
all Bethe roots. This is trivially possible for finite Trotter number,
but also in Trotter limit $N \rightarrow \infty$ \cite{GoSu10}. Then
the right hand side of (\ref{newg}) is holomorphic in $\la$ for
$|\Im \la| < \g/2$. It follows that
\begin{equation} \label{limg1}
     \lim_{\Re \la \rightarrow \infty} G(\la , \n) =
        q^{-\a} - \r(\n) - (q^\a - q^{- \a}) \int_C \rd m(\m) G(\m, \n)
	= 0 \epp
\end{equation}
Here the second equation will appear as lemma \ref{lem:onep} below.
We postpone the proof, because it needs some preparation.

For the calculation of the asymptotics of $G$ for large $\Re \n$ we have
to take into account that $G(\la, \n)$ as a function of $\la$ has pole
at $\la = \n$ with residue $- \r (\n)$. Hence, for $\n$ outside $C$,
\begin{multline} \label{newgcont}
     G(\la, \n) =  q^{-\a} \cth(\la - \n - \h) - \r (\n) \cth (\la - \n) \\
                   - \frac{K_\a (\la - \n)}{1 + \fa(\n |\k)}
                   + \int_{C} \rd m(\m) K_\a (\la - \m) G(\m, \n) \epp
\end{multline}
Using that
\begin{equation} \label{limrhoa}
     \lim_{\Re \n \rightarrow \infty} \r (\n) =
        \frac{q^{\k + \a} + q^{- \k - \a}}{q^\k + q^{- \k}} \epc \qd
     \lim_{\Re \n \rightarrow \infty} \fa (\n |\k) = q^{- 2 \k}
\end{equation}
and setting $g(\la) = \lim_{\n \rightarrow \infty} G(\la, \n)$ we obtain
from (\ref{newgcont})
\begin{equation}
     g(\la) = \int_C \rd m(\m) K_\a (\la - \m) g(\m) \epp
\end{equation}
Then $g(\la) = 0$, and (\ref{limg}) is proved.

A similar argument can be applied to prove (\ref{asymppsi2}). For $\n_2$
outside $C$ we have
\begin{multline} \label{Psicont2}
     \Ps (\n_1, \n_2) = \int_{C} \rd m(\m) G(\m, \n_2)
        \bigl(q^\a \cth(\m - \n_1 - \h) - \r(\n_1) \cth(\m - \n_1) \bigr) \\
          - \frac{1}{1 + \fa (\n_2 |\k)} \bigl( q^\a \cth (\n_2 - \n_1 - \h)
          - \r (\n_1) \cth(\n_2 - \n_1) \bigr) \epp
\end{multline}
Using the second equation (\ref{limg}) and (\ref{limrhoa}) we obtain
(\ref{asymppsi2}).

For the proof of (\ref{asymppsi1}) we note that
\begin{multline} \label{Psicont1}
     \Ps (\n_1, \n_2) = \int_{C} \rd m(\m) G(\m, \n_2)
        \bigl(q^\a \cth(\m - \n_1 - \h) - \r(\n_1) \cth(\m - \n_1) \bigr) \\
          - \frac{G(\n_1, \n_2)}{1 + \fa (\n_1 |\k)}
\end{multline}
if $\n_1$ is outside $C$. Using (\ref{limg1}) we conclude that
\begin{equation} \label{asymppsig1}
     \lim_{\n_1 \rightarrow \infty} \Ps (\n_1, \n_2) =
        - \frac{q^\k - q^{- \k}}{q^\k + q^{- \k}} 
          \lim_{\n_1 \rightarrow \infty} G (\n_1, \n_2)
        - \frac{q^{- \a} - \r (\n_2)}{1 + q^{2 \k}} \epp
\end{equation}
Thus, (\ref{asymppsi1}) follows by means of (\ref{limg}).
\end{proof}

Here a few comments are in order. For the correlation functions of the
XXZ chain, \cite{JMS08,BoGo09} it is actually not the function $\Ps$ but the
closely related function $\om$ \footnote{We follow here the notation of
\cite{BJMS09bpp}. In ref.\ \cite{BoGo09} a slightly different notation for
$\om$ was used.} which is at the heart of the theory
\begin{multline}
     \om(\n_1, \n_2) = 
        2 \Ps (\n_1, \n_2)\re^{\a (\n_1 - \n_2)} +\\ 
        4\, \bigl( \bigl(1+\rho(\nu_1)\rho(\nu_2) \bigr)g(\xi)
	   -\rho(\nu_1)g(q^{-1}\xi)-\rho(\nu_2)g(q\xi)\bigr)
%        \dst{\2} K_\a (\n_1 - \n_2)
%        + \bigl(\r (\n_1) - \r (\n_2)\bigr) \cth(\n_1 - \n_2) \bigr]
%          \re^{\a (\n_1 - \n_2)}
\end{multline}
where $g(\xi)=\Delta_{\xi}^{-1}\psi(\xi)$ with $\xi=\re^{ \n_1 - \n_2}$
should be understood as in (2.10) of ref.\ \cite{BJMS09bpp}. Using the
explicit form of the function $\psi(\xi)=\frac{\xi^{\a}}{2}
\frac{\xi^2+1}{\xi^2-1}$ and the above lemma (2.1) one can see that 
\[
\lim_{\nu_1\rightarrow\pm\infty}\re^{-\a (\n_1 - \n_2)}\om(\n_1,\n_2)=0,\quad
\lim_{\nu_2\rightarrow\pm\infty}\re^{-\a (\n_1 - \n_2)}\om(\n_1,\n_2)=0\epp
\]
This is one of normalization conditions for $\om$ introduced in \cite{JMS08}.
%The function $\Ps$ is only the least
%well understood part of $\om$. 

The asymptotics of $G$ and $\Psi$ with
respect to the first argument was derived in \cite{BoGo09}. The reasoning
there was also based on the second equation (\ref{limg1}), which was
obtained rather indirectly by means of the reduction property of the density
matrix and a multiple integral representation for the six-vertex model
with disorder parameter. Below we shall present a more direct proof of
it in lemma \ref{lem:onep}, based on a representation of the generalized
dressed charge in terms of $Q$-functions.

Perhaps the simplest proof of lemma \ref{lem:onep} utilizes the symmetry
\cite{JMS08,BoGo09}
\begin{equation}
     \Ps (\n_1, \n_2|\k, \a) = \Ps (\n_2, \n_1|- \k, - \a) \epp
\end{equation}
If we combine this with $\r (\la|\k, \a) = \r (\la|- \k, - \a)$, then
(\ref{asymppsi1}) follows from (\ref{asymppsi2}). But (\ref{asymppsi1})
inserted into (\ref{asymppsig1}) implies $\lim_{\n_1 \rightarrow \infty}
G (\n_1, \n_2) = 0$, and lemma \ref{lem:onep} follows with the first
equation (\ref{limg1}).

Still, this is a little indirect and unsatisfactory. Here we are going
for a more direct proof based upon the properties of the dressed charge
function defined by
\begin{equation} \label{newsigma}
     \s (\la) = 1 + \int_C \rd m(\m) \s(\m) K_\a (\m - \la) \epp
\end{equation}

\begin{lemma} \label{lem:sigma}
Dressed charge in terms of $Q$-functions.
\begin{equation} \label{sigmaphi}
     \s(\la) = \frac{q^\a \phi (\la - \h) - q^{- \a} \phi (\la + \h)}
                    {\phi_0 (q^\a - q^{- \a})} \epc
\end{equation}
where
\begin{equation} \label{defssigmaphi}
     \phi (\la) = \frac{Q(\la |\k + \a)}{Q(\la |\k)} \epc \qd
     \phi_0 = \ch \biggl( \int_C \frac{\rd \la}{2 \p \i}
                              \ln \biggl( \frac{1 + \fa (\la |\k + \a)}
                                               {1 + \fa (\la |\k)} \biggr)
                          \biggr) \epp
\end{equation}
\end{lemma}

\begin{proof}
We recall from the appendix of \cite{BoGo09} that
\begin{equation}
     \phi (\la) = \frac{Q(\la |\k + \a)}{Q(\la |\k)}
                = \prod_{j=1}^{N/2} \frac{\sh(\la - \la_j (\k + \a))}
                                         {\sh(\la - \la_j (\k))} \epc
\end{equation}
where the $\la_j$ are the Bethe roots of the dominant eigenstate of
the quantum transfer matrix which are located inside the contour $C$.
Due to the $t$-$Q$-equation (\ref{tq})
\begin{equation} \label{id0}
     \phi (\la) = \frac{q^{-\a} \phi (\la + \h)}{\r (\la)}
                  + \frac{q^\a \phi (\la - \h) - q^{- \a} \phi (\la + \h)}
                         {\r(\la) (1 + \fa (\la |\k))} \epp
\end{equation}
Here the first term on the right hand side is holomorphic inside $C$.
Hence,
\begin{equation} \label{id1}
     \int_C \rd m(\m)
        \bigl( q^\a \phi (\m - \h) - q^{- \a} \phi (\m + \h) \bigr)
        K_\a (\m - \la) =
        \int_C \frac{\rd \m}{2 \p \i} \, \phi (\m) K_\a (\m - \la) \epp
\end{equation}
The latter integral can be calculated. Note that
\begin{equation} \label{period}
     \phi (\la + \i \p) = \phi (\la) \epc \qd
     K_\a (\la + \i \p) = K_\a (\la)
\end{equation}
and
\begin{equation} \label{asymptotics}
     \lim_{\Re \la \rightarrow \pm \infty} \phi (\la) = b^{\pm 1} \epc
        \qd b = \exp \Bigl(
                \sum_{j=1}^{N/2} \bigl(\la_j (\k) - \la_j (\k + \a)\bigr)
                     \Bigr) \epp
\end{equation}
With the contours sketched in figure \ref{fig:kontur} it follows that
\begin{multline} \label{id2}
     \int_{\tilde C + C} \frac{\rd \m}{2 \p \i} \, \phi (\m) K_\a (\m - \la)
        \\ =
     \int_{I_2 + I_4} \frac{\rd \m}{2 \p \i} \, \phi (\m) K_\a (\m - \la)
        = - \frac{b + b^{-1}}{2} (q^\a - q^{- \a}) \\[1ex]
     = q^{- \a} \phi (\la + \h) - q^\a \phi (\la - \h) +
       \int_C \frac{\rd \m}{2 \p \i} \, \phi (\m) K_\a (\m - \la) \epp
\end{multline}
Here we have used the periodicity (\ref{period}) in the first equation,
the fact that the integral is independent of $R$ and the asymptotics
(\ref{asymptotics}) in the second equation, and the fact that
$\phi$ is free of poles inside $\tilde C$ in the third equation.
Setting $\phi_0 = (b + b^{- 1})/2$ and combining (\ref{id1}) and
(\ref{id2}) we obtain (\ref{sigmaphi}).
\begin{figure}
    \centering
    \includegraphics[height=60.mm,angle=0,clip=true]{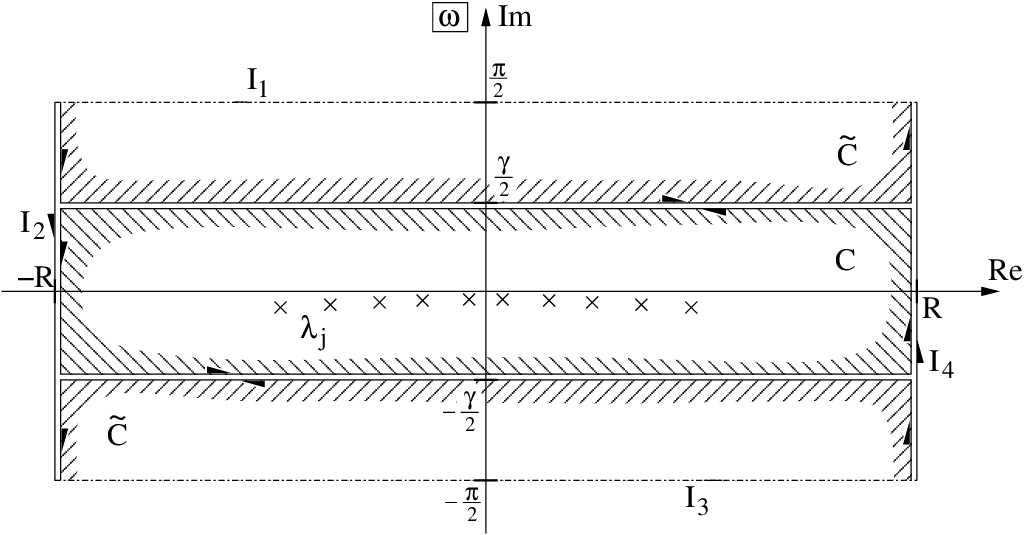}
    \caption{\label{fig:kontur} Contours used in the proofs of lemma
             \ref{lem:sigma} and lemma \ref{lem:onep}.}
\end{figure}

It remains to show the second equation (\ref{defssigmaphi}). It follows
from
\begin{equation}
     \int_C \frac{\rd \la}{2 \p \i} \:
          \ln \bigl(1 + \fa (\la|\k)\bigr)
        = - \int_C \frac{\rd \la}{2 \p \i} \: \la \, \6_\la
          \ln \bigl(1 + \fa (\la|\k)\bigr)
        = - \sum_{j=1}^{N/2} \la_j (\k) - \be/2 \epc
\end{equation}
and the proof is complete.
\end{proof}
Clearly (\ref{sigmaphi}) turns into (\ref{sigmazero}) in the limit $\a
\rightarrow 0$. Equipped with lemma \ref{lem:sigma} we can now proceed
with proving
\begin{lemma} \label{lem:onep}
An identity for one-point functions.
\begin{equation}
     \int_C \rd m(\la) G(\la, \n)
        = \frac{q^{- \a} - \r(\n)}{q^\a - q^{- \a}} \epp
\end{equation}
\end{lemma}
\begin{proof}
First of all applying the dressed function trick to (\ref{newg}) and
(\ref{newsigma}) we obtain
\begin{equation}
     \int_C \rd m(\la) G(\la, \n)
        = \int_C \rd m(\la) \s (\la) \bigl(
          q^{-\a} \cth(\la - \n - \h) - \r (\n) \cth (\la - \n) \bigr) \epp
\end{equation}
Then we insert (\ref{sigmaphi}) and (\ref{id0}) into the right hand side.
It follows that
\begin{align}
     & \int_C \rd m(\la) G(\la, \n) = \frac{1}{\phi_0 (q^\a - q^{- \a})}
        \notag \\[1ex] & \mspace{36.mu} \int_C \frac{\rd \la}{2 \p \i}
        \biggl(\phi (\la) - \frac{q^{- \a} \phi (\la + \h)}{\r (\la)}
               \biggr) \bigl( q^{-\a} \cth(\la - \n - \h)
                              - \r (\n) \cth (\la - \n) \bigr) =
        \displaybreak[0] \notag \\
        & = \frac{1}{\phi_0 (q^\a - q^{- \a})} \biggl(
          q^{- \a} \phi(\n + \h) \notag \\ & \mspace{135.mu}
          + \int_{I_2 + I_4 - \tilde C} \frac{\rd \la}{2 \p \i} \:
            \phi (\la) \bigl( q^{-\a} \cth(\la - \n - \h)
                              - \r (\n) \cth (\la - \n) \bigr) \biggr)
        \notag \\ & = \frac{q^{- \a} - \r (\n)}{q^\a - q^{- ßa}} \epp
\end{align}
Here we have again referred to figure \ref{fig:kontur} in the second
equation.
\end{proof}

\section{Conclusions}
We would like to conclude with three more remarks. First, all the above
remains valid if we consider a more general six-vertex model with a more
general inhomogeneous choice of parameters in vertical direction. In that
case the functions $a(\la)$ and $d(\la)$ in (\ref{detpara}) and also the
driving term in (\ref{nlien}) have to be modified as in \cite{BoGo09}.

Second, in reference \cite{BoGo09} the calculation of the limit
$\lim_{\n_1 \rightarrow \infty} \Ps (\n_1, \n_2)$ was the only point,
where we had to resort to a multiple integral representation, when we
showed that the functions $\om$ as defined in references \cite{BoGo09}
and \cite{JMS08} are identical. With the proof presented in this note
the approach to the correlation functions of the XXZ chain, based on
the discovery of a hidden Grassmann symmetry, as developed in
\cite{BJMST08a,JMS08,BoGo09} becomes logically independent of the
multiple integral representation that was also obtained in \cite{BoGo09}.

Third, the function $\phi$ used above to express $\s$ seems quite
interesting and may deserve further attention. Is it possible to express
the solutions of other linear integral equations in terms of $\phi$?
And is there a useful integral equation for $\phi$ itself? We hope we
can come back to these questions in the future.

\section*{Acknowledgment}
It a great pleasure and honour to dedicate this note to our dear friend
Tetsuji Miwa on the occasion of his 60th birthday. We are grateful to Fedor
Smirnov for suggesting us to prove lemma \ref{lem:onep} directly. We wish
to thank Alexander Seel for drawing the figure for us. Our work was
generously supported by the Volkswagen foundation.

%\bibliographystyle{ourbook}
%\bibliography{hub}

\end{document}